\documentclass[10pt]{article}
\usepackage{dcolumn}
\usepackage{bm}
\usepackage{verbatim}       

\usepackage[pdftex]{graphicx}
\usepackage{amsthm}
\usepackage{amssymb}
\usepackage{undertilde}

\usepackage{amstext}
\usepackage{amsmath}
\usepackage{amsfonts}
\usepackage{units}
\usepackage{mathrsfs}
\usepackage{cite}
\newcommand{\p}{\partial}

\newcommand{\dd}{{\rm d}}

\newcommand{\bd}{\begin{definition}}                
\newcommand{\ed}{\end{definition}}                  
\newcommand{\bc}{\begin{corollary}}                 
\newcommand{\ec}{\end{corollary}}                   
\newcommand{\bl}{\begin{lemma}}                     
\newcommand{\el}{\end{lemma}}                       
\newcommand{\bp}{\begin{proposition}}            
\newcommand{\ep}{\end{proposition}}                
\newcommand{\bere}{\begin{remark}}                  
\newcommand{\ere}{\end{remark}}                     

\newcommand{\bt}{\begin{theorem}}
\newcommand{\et}{\end{theorem}}

\newcommand{\be}{\begin{equation}}
\newcommand{\ee}{\end{equation}}

\newcommand{\bit}{\begin{itemize}}
\newcommand{\eit}{\end{itemize}}
\newtheorem{theorem}{Theorem}[section]
\newtheorem{corollary}[theorem]{Corollary}
\newtheorem{lemma}[theorem]{Lemma}
\newtheorem{proposition}[theorem]{Proposition}
\theoremstyle{definition}
\newtheorem{definition}[theorem]{Definition}
\theoremstyle{remark}
\newtheorem{remark}[theorem]{Remark}

\begin{document}




\title{Lorentzian manifolds properly isometrically embeddable in Minkowski spacetime}

\author{E. Minguzzi\thanks{
Dipartimento di Matematica e Informatica ``U. Dini'', Universit\`a
degli Studi di Firenze, Via S. Marta 3,  I-50139 Firenze, Italy.
E-mail: ettore.minguzzi@unifi.it. } }

\date{}

\maketitle

\begin{abstract}
\noindent I characterize the Lorentzian manifolds properly isometrically embeddable in Minkowski spacetime (i.e.\ the  Lorentzian submanifolds of Minkowski spacetime that are also closed subsets). Moreover, I prove that the Lorentzian manifolds that can be properly conformally embedded in Minkowski spacetime coincide with the globally hyperbolic spacetimes. Finally, by taking advantage of the embedding, I obtain an infinitesimal version of the distance formula.
\end{abstract}

\section{Introduction}

In  this work we investigate the problem of characterizing the topologically closed connected Lorentzian   submanifolds of Minkowski spacetime or their conformal structure.

It is convenient to start by introducing some terminology and notations. Our convention for the Lorentzian signature is $(-,+,\cdots,+)$. A {\em spacetime } is a connected time-oriented Lorentzian manifold (second countable and Hausdorff) of dimension $n$.
A vector $v$ on a spacetime $(M,g)$ is {\em causal} or {\em nonspacelike} if $g(v,v)\le 0$ and $v\ne 0$, and {\em timelike} if the strict inequality holds. We might write, for $v$ causal vector,  $\Vert v\Vert_g:=\sqrt{-g(v,v)}$. The $N+1$-dimensional Minkowski spacetime is denoted $\mathbb{L}^{N+1}$ and its metric is $\eta^{(N+1)}$.

We recall  that a $C^1$ function $f\colon M\to \mathbb{R}$ is called {\em temporal} if it has past-directed timelike gradient, or equivalently if $\dd f$ is positive over the future-directed causal vectors,   while it is {\em steep} if it satisfies $\dd f(v)\ge \sqrt{-g(v,v)}$ for every future-directed causal vector $v$ ({\em strictly steep} if the inequality is strict), or equivalently \cite[Thm.\ 1.23]{minguzzi18b}, if $f$ is temporal and $\Vert \nabla^g f\Vert_g\ge 1$. It is $h$-steep if $\dd f(v)\ge \sqrt{h(v,v)}$ for every future-directed causal vector $v$, where $h$ is a Riemannian metric.

A spacetime is {\em stably causal} if the cones can be widened by preserving the causality condition (i.e.\ the absence of closed causal curves), or equivalently, if the property of causality is stable in the  $C^0$ topology on  metrics \cite{hawking73}. A spacetime admits a temporal function iff it is stably causal \cite{bernal04b,bernal04,fathi12,chrusciel13,minguzzi17,bernard18}.

A spacetime is {\em stable} if both  causality and the finiteness of the Lorentzian distance are stable in the  $C^0$ topology on  metrics \cite{minguzzi17}.

It is clear that a connected Lorentzian manifold embedded in Minkowski spacetime inherits a time orientation from that of Minkowski spacetime. 
   Thus, without loss of generality, we can restrict ourselves to the problem of studying the embeddings of spacetimes into Minkowski spacetime.

The problem of characterizing the spacetimes $(M,g)$ isometrically embeddable in Minkowski spacetime $\mathbb{L}^{N+1}$ for some $N$, has been fully solved in recent years (for the analogous embedding problem  in generic pseudo-Riemannian manifolds see \cite{clarke70}).

We recall that an embedding is a map $\phi: M\to  \mathbb{L}^{N+1}$ which is a homeomorphism onto its image such that $\phi_*$ is injective. It is an isometry if $g=\phi^* \eta^{(N+1)}$, and a conformal isometry if $g=  \Omega^2 \phi^* \eta^{(N+1)}$, for some function $\Omega: M\to \mathbb{R}$.

By a result due to Whitney every $C^1$ manifold admits a  unique smooth compatible structure
(Whitney \cite{whitney36,whitney44})  \cite[Thm.\ 2.9]{hirsch76} thus in this type of results, as $M$ is tacitly assumed to be $C^{k+1}$, $k\ge 0$ (otherwise $C^k$ metrics do not make sense), the regularity of $M$ is  often not mentioned as it can be promoted to $C^\infty$.

\begin{theorem} \label{bys}
Let $(M,g)$ be an $n$-dimensional spacetime $(M,g)$, where $g$ is $C^k$, $k\in \mathbb{N}\backslash\{1,2\} \cup \{\infty\}$, and let $s:=k$ for $k\ge 3$; $s:=1$ for $k=0$.  The next properties are equivalent:
\begin{itemize}
\item[(a)] there is  a $C^s$  isometric embedding in Minkowski spacetime $\mathbb{L}^{N+1}$ for some $N\ge 1$,
\item[(b)] $(M,g)$ admits a $C^s$ steep  function,
\item[(c)]  $(M,g)$ is stable.
 \end{itemize}
 In this case $N$ can be chosen to be $N(n,k)$, the optimal value  for the analogous Riemannian problem.
\end{theorem}
Characterization (b) was proved by M\"uller and Sanchez in \cite[Thm.\ 1.1]{muller11}  for $k\ge 3$, while (c) was proved by the author in \cite[Thm.\ 3.10, 4.13]{minguzzi17}.  In (b) $C^s$ can be replaced by $C^1$. The spacetimes characterized by this theorem are also those for which the distance formula holds true \cite[Thm.\ 4.6]{minguzzi17}. The globally hyperbolic spacetimes admit a steep   function and as such are isometrically embeddable in Minkowski spacetime \cite{muller11,minguzzi16a} \cite[Thm.\ 3.12]{minguzzi17}, but the class of stable spacetimes is much larger. For instance, the distinguishing spacetimes that admit a finite and continuous Lorentzian distance are stable \cite[Cor.\ 4.1]{minguzzi17}.

\begin{remark}
The last statement of the theorem and also the case $k=0$ require some comment. The classical result by Nash on the isometric embedding of Riemannian manifolds \cite{nash54,nash56} \cite[Cor.\ 30, ]{delellis19} \cite{gromov70} is (the case $k=0$ is named after Nash and Kuiper):

\begin{theorem}
Let $\Sigma$ be a $C^\infty$  $n$-dimensional  non-compact manifold endowed with a $C^k$ metric $h$, $k\in \mathbb{N}\backslash\{1,2\} \cup \{\infty\}$, and let $s:=k$ for $k\ge 3$; $s:=1$ for $k=0$.  There is a $C^s$  isometric embedding $\varphi: \Sigma\to E^N$ for some $N>0$.
\end{theorem}
The optimal value for $N$ will be denoted $N(k,n)$ and for $\Sigma$ non-compact it is known to satisfy $N(k,n)\le \frac{1}{2} (n+1) (n (3 n+11)+4)$, but the actual bound will not be important in what follows.


 As we shall also later recall  in more detail,  the proof of the Lorentzian version \cite{muller11} clarifies that in the Lorentzian case the embedding can be found of the form $\phi=(\sqrt{2} t,\varphi)$ where $\varphi$ is the Nash embedding of a related Riemannian metric $h$ on $M$ while $t$ is the smooth function in Thm.\ \ref{bys}(b). Thus the Lorentzian embedding has the same regularity of the Nash embedding and the ambient manifold can be chosen to be $\mathbb{L}^{N(n,k)+1}$.


The $k=0$  case was not comprised in the analysis of \cite{muller11}.
In \cite{minguzzi17} we proved that stable spacetimes with $C^0$ metrics admit smooth steep functions, and so the Nash-Kuiper theorem can be used to find $C^1$ embeddings in $\mathbb{L}^{N(n,0)+1}$.
\end{remark}

The problem of the {\em conformal} embedding is also understood.
The following result is proved in \cite[Cor.\ 1.4]{muller11}, see \cite[Thm.\ 2.62]{minguzzi17} for the $k=0$ case.

\begin{theorem} \label{bys2}
Let $(M,g)$ be an $n$-dimensional spacetime, where $g$ is $C^k$, $k\in \mathbb{N}\backslash\{1,2\} \cup \{\infty\}$, and let $s:=k$ for $k\ge 3$; $s:=1$ for $k=0$.   The next properties are equivalent
\begin{itemize}
\item[(a)] there is  a $C^s$ conformal embedding in Minkowski spacetime $\mathbb{L}^{N+1}$ for some $N\ge 1$,
\item[(b)] $(M,g)$ admits a $C^s$ temporal function,
\item[(c)]  $(M,g)$ is stably causal.
 \end{itemize}
 In this case $N$ can be chosen to be $N(n,k)$, the optimal value  for the isometric Riemannian problem.
\end{theorem}
For the proof one shows that the spacetime becomes stable after a suitable conformal rescaling of the metric \cite[Thm.\ 2.62]{minguzzi17}.
Once again $C^s$ in (b) can be replaced by $C^1$.

\begin{remark} \label{jvo}
Let $[g]$ denote the conformal class of $g$.
What are the spacetimes $(M,g)$ such that for every $\tilde g\in [g]$, $(M,\tilde g)$ is isometrically embeddable in Minkowski spacetime? Any isometrically embeddable spacetime can be easily shown to be stably causal with finite Lorentzian distance  \cite{muller11}. It is known that if the spacetime is stably causal and the Lorentzian distance is finite for every element of the conformal class,  then the spacetime is globally hyperbolic  \cite[Thm.\ 4.30]{beem96}. Conversely, given a globally hyperbolic spacetime, every element in the conformal class is globally hyperbolic (for the causal structure is independent of the conformal factor) thus isometrically embeddable (because the globally hyperbolic spacetimes admit steep time function \cite{muller11,minguzzi16a}). The answer to the question is then: the globally hyperbolic spacetimes \cite{muller11}.
\end{remark}

In this work we are going to consider the analogous problems for {\em proper} isometric/conformal emeddings in Minkowski spacetime. We recall that a continuous mapping is said to be {\em proper} if the inverse
images of compact sets are compact. For an embedding this condition is equivalent to the topological closedness of the image \cite{gromov70} (i.e.\ no boundary point) a fact that will be used without further mention in what follows. For this reason we shall also speak of {\em closed} embeddings. In other words we shall be interested in the characterization of the closed Lorentzian submanifolds of Minkowski spacetime, and on the conformal structures that can be similarly embedded in Minkowski spacetime.



String theorists make current use of the notion of embedded spacetime, though they refer to them as {\em branes} \cite{moore05}. Embedded spacetimes with no boundary seem to be very natural objects. Near a boundary point the submanifold can oscillate wildly, which is why the Lorentzian submanifolds of Minkowski with boundary get identified with the large category of stable spacetimes. Much nicer spacetimes are expected with the imposition of the no boundary condition.

Among the nicest spacetimes we find the globally hyperbolic ones. However, M\"uller has shown that there are simple globally hyperbolic spacetimes that are not properly isometrically embeddable in Minkowski \cite[Example 1]{muller13}. As a consequence, there must exist some  category of spacetimes which is  more restrictive than that of globally hyperbolic spacetimes.

Still we face a problem: can the spacetimes properly isometrically embeddable in Minkowski be characterized through intrinsic properties?
In this work we are going to prove that they can, so pointing to a new category of spacetimes that might play a significant role in physics.

\subsection{Proper embeddings in Euclidean space}

Let us first consider the Riemannian case.
In \cite[p. 11--12]{gromov70} Gromov and Rokhlin stated that the Nash isometric embedding can be found proper if the Riemannian manifold is complete.
A clear and simple proof of this result was obtained by  M\"uller \cite{muller09}.
Let us rephrase it emphasizing the arguments that shall be useful in what follows.

\begin{theorem} \label{arg}
Let $(S,h)$ be a  Riemannian manifold with  $h\in C^k$, $k\in \mathbb{N}\backslash\{1,2\} \cup \{\infty\}$, and let $s:=k$ for $k\ge 3$; $s:=1$ for $k=0$.
It is properly isometrically $C^s$  embeddable  in $E^N$ for some $N$ iff there is a $C^1$ proper function $r\colon S \to (0,+\infty)$ such that\footnote{Here $\Vert \nabla^h r \Vert_h:=h^{-1}(\dd r, \dd r)$ thus the expression makes sense for a $C^0$ metric $h$.} $\Vert \nabla^h r \Vert_h<1$. In the last sentence $C^1$ can be replaced by {\em smooth}.
\end{theorem}

We recall that a function is proper if the inverse image of any compact set is compact. Observe that it could be $\nabla^h r= 0$ at some point.

\begin{proof}
For the last statement, by \cite[Thm.\ 2.6]{hirsch76}  $C^\infty(S, \mathbb{R})$ is dense in $C^1(S, \mathbb{R})$ endowed with the Whitney
strong topology \cite[p. 35]{hirsch76}, thus we can find $r'\in  C^\infty(S, \mathbb{R})$, which approximates $r$, up to the
first derivative, as accurately as we want over $S$, hence in such a way that $\vert r'-r\vert<1$ and $\Vert \nabla^h r \Vert_h<1$, the former inequality implying that $r'$ is also proper.

$\Rightarrow$. Let $\{e_i\}$ be an orthonormal basis for $E^N$ and let  $\{X^i\}$  be the associated Cartesian coordinates of $E^N$. The manifold $S$ can be regarded as a Riemannian submanifold with induced metric $h$.  Let
\[
R=\sqrt{1+\vec{X}\cdot \vec{X}},
\]
and let $r=R\vert_S$, so that $r$ is $C^s$ hence $C^1$. Observe that
\[
\nabla^E R=\sum_i \frac{X^i}{R} e_i,
\]
thus $\Vert  \nabla^E R\Vert_E= \sqrt{\frac{\vec{X}^2}{1+\vec{X}^2}}<1$, thus for every $v\in TE^N$, $v \ne \vec{0}$, we have by the Cauchy-Schwarz inequality
\[
\vert \p_v R \vert=\vert \nabla^E R\cdot_E v\vert <\Vert v \Vert_E.
\]
For every $v\in TS$, $v\ne \vec{0}$, we have as a consequence, $\vert\p_v r \vert <\Vert v\Vert_h$. If $\nabla^h r\ne 0$, then with $v=\nabla^h r$,
\[
\vert h^{-1}(\dd r,\dd r)\vert=\vert \dd r(v)\vert<\Vert \nabla^h r\Vert_h=\sqrt{h^{-1}(\dd r,\dd r)},
\]
thus $\Vert \nabla^h r \Vert_h=\sqrt{h^{-1}(\dd r,\dd r)}<1$. Moreover, $r$ is proper because for $a>0$, $r^{-1}([0,a])=R^{-1}([0,a])\cap S$. This set is closed because $R$ is continuous and $S$ is closed. Moreover, it is contained in the ball $B_E(0,a)$ thus it is compact.

$\Leftarrow$. We can assume that $r$ is smooth. Consider the $C^k$ metric $\tilde h=h-\dd r \otimes \dd r$. We want to show that it is Riemannian. This is clear at those points where $\dd r=0$, so we can just focus on points where $\dd r\ne 0$. On the tangent spaces to the level sets of $r$ it is a positive quadratic form thus, in order to show that it is Riemannian, we need only to prove that $\tilde h(\nabla^h r, \nabla^h r)>0$ because
\[
\tilde h(\nabla^hr,\cdot)=h(\nabla^h r, \cdot)-h^{-1}(\dd r,\dd r )\, \dd r=\big(1-h^{-1}(\dd r,\dd r)\big)\, \dd r
\]
so that the diagonal terms vanish. Now
\[
\tilde h(\nabla^h r, \nabla^h r)=\big(1-h^{-1}(\dd r,\dd r)\big) h^{-1}(\dd r,\dd r),
\]
which is positive because $0<h^{-1}(\dd r, \dd r)=\Vert \nabla^h r\Vert_h^2<1$. Thus by Nash's theorem there is, for some $N>0$, a $C^s$ isometric embedding $\tilde \varphi\colon  S \to E^{N-1}$ such that $\tilde \varphi^* \delta^{(N-1)}=\tilde h$, where $\delta^{(N-1)}$ is the Euclidean metric in $E^{N-1}$. Thus $\varphi\colon S\to E^N$, $\varphi=(\tilde \varphi, r)$
\[
\varphi^* \delta^{(N)}=\tilde h+\dd r\otimes \dd r=h.
\]
The embedding  cannot have a boundary point otherwise the coordinate $X^{N}$ would be bounded in a neighborhood of it, which would imply that we could find $p_n\in S$, $p_n \to \infty$, with $r(p_n)$ bounded in contradiction with the properness of $r$.
\end{proof}

\begin{theorem}
Let $(S,h)$ be a  Riemannian manifold with  $h\in C^k$, $k\in \mathbb{N} \cup \{\infty\}$.
The  Riemannian manifold $(S,h)$ admits a $C^1$ (equiv.\ $C^\infty$) proper function $r\colon S\to [0,+\infty)$ such that $\Vert \nabla^h r\Vert_h<1$ iff it is complete.
\end{theorem}

In the proof with {\em complete} we shall mean the Heine-Borel property: bounded closed subsets are compact. Nevertheless, all standard Hopf-Rinow equivalences familiar from the $C^2$ metric theory are preserved even  for $C^0$ metrics. Indeed, in this case  we have at our disposal the Hopf-Rinow-Cohn-Vossen's theorem which holds for locally compact length spaces. The Riemannian spaces with continuous metric are of this type \cite{burtscher15}.

\begin{proof}
$\Rightarrow$. Let $\sigma\colon [0,1] \to S$ be a regular $C^1$ curve connecting $p$ to $q$. We have
\[
\vert r(q)-r(p) \vert= \vert \int_0^1 \dd r(\dot \sigma) \dd t \vert\le \int_0^1  \vert \nabla^h r \cdot_h \dot \sigma \vert \dd t \le   \int_0^1  \Vert  \dot \sigma \Vert_h \dd t=\ell_h(\sigma)
\]
Thus, taking the infimum over $\sigma$ we get $\vert r(q)-r(p)\vert \le d^h(p,q)$. We conclude that the closed balls are compact since $d$ is continuous and $r$ is proper, hence $(M,S)$ is complete.

$\Leftarrow$.
Let $o\in S$ and let $f:=d(o,\cdot)$.
 By the triangle inequality $\vert f(q)- f(p)\vert \le d(p,q)$
 thus $f$ is 1-Lipschitz.
 By \cite{azagra07} we can find a smooth 2-Lipschitz function $g$ (thus $\Vert \nabla^h g\Vert_h\le 2$) such that $\vert g-f\vert \le 1$. Then the function $r=g/4$ is smooth proper and such that $\Vert \nabla^h r\Vert_h \le 1/2 <1$.
\end{proof}

We arrive at the theorem in \cite{muller09} improved with details on the regularity.

\begin{corollary} \label{odw}
Let $(S,h)$ be a  Riemannian manifold with  $h\in C^k$, $k\in \mathbb{N}\backslash\{1,2\} \cup \{\infty\}$, and let $s:=k$ for $k\ge 3$; $s:=1$ for $k=0$. It is properly isometrically $C^s$  embeddable in  Euclidean space $E^N$ for some $N$ iff it is complete.
\end{corollary}

The proof shows that $N$ can be chosen to be $N(n,k)+1$, where $N(n,k)$ is the optimal value for the non-closed isometric embedding.

\begin{remark}
A complete Riemannian manifold might admit non-closed isometric embeddings into Euclidean space. For instance,
$\mathbb{R}$ admits non-closed isometric embeddings into $\mathbb{R}^2$: consider a curve spiraling to the origin of $\mathbb{R}^2$ or to its unit circle. In order to select a closed embedding the trick was to choose a proper  function as one coordinate for the embedding.


\end{remark}

\subsection{Proper embeddings in Minkowski spacetime}

 If we add a time coordinate to the example of the previous remark, we see that 1+1 Minkowski spacetime admits non-closed isometric embeddings into 2+1 Minkowski spacetime  though it is a well behaved spacetime. In order to prove that a spacetime $(M,g)$ admits a closed isometric embedding in Minkowski we need at least two functions on $M$, one function $t$ that goes to infinity in the timelike direction and the other function $r$ that goes to infinity in the spacelike direction.
Thus we need to obtain a Riemannian metric as follows
\[
g+\dd t^2-\dd r^2.
\]
If we were not interested on the properness of the embedding we would just try to find a Riemannian metric by using a Cauchy temporal function $t$ as follows
\[
g+\dd t^2
\]
which, contracting twice with $\nabla t=g^{-1}(\dd t, \cdot)$, gives
\[
g^{-1}(\dd t, \dd t )+g^{-1}(\dd t, \dd t )^2=g^{-1}(\dd t, \dd t ) [1+g^{-1}(\dd t, \dd t )] .
\]
Thus, since this has to be positive, we get $-g^{-1}(\dd t, \dd t )>1$ which is the strict steepness condition. Then one observes that on $\textrm{ker}\, \dd t$ the metric is also positive, and that the diagonal term vanish because contraction with $\nabla t$ just on the left gives
\[
[1+g^{-1}(\dd t, \dd t )]\, \dd t
\]
which vanishes on $\textrm{ker}\, \dd t$. This is, essentially, the strategy followed by M\"uller and S\'anchez in \cite{muller11}. Unfortunately, by working with just one function one cannot force the embedding to be closed, as the example of the above embedding of Minkowski  1+1 in Minkowski 2+1 shows.

We shall need the following results.

\begin{lemma} \label{nis}
Let $g$ a Lorentzian  bilinear form and let $u$ and $v$ be vectors such that $g(u,u),g(v,v)\ge 0$. If the inequality
\begin{equation} \label{bud}
[1+g(u,v)]^2\le g(u,u)\, g(v,v)
\end{equation}
holds, then $u$ and $v$ are causal and of the same time orientation.
If the inequality is strict then they are timelike.
\end{lemma}

\begin{proof}
Neither $u$ nor $v$ can vanish, for we would get $1\le 0$, thus they are causal. From the reverse Cauchy-Schwarz inequality $ g(u,u)\, g(v,v)\le  g(u,v)^2$ (which is independent of the relative time orientation) and (\ref{bud}) we get $g(u,v)\le -1/2<0$, thus $u$ and $v$ have the same time orientation.
It is clear that if the inequality is strict then neither of them can be lightlike, for we would get $0\le [1+g(u,v)]^2<0$.
\end{proof}

\begin{lemma} \label{llo}
Let $g$ be a Lorentzian bilinear form and let $\alpha$ and $\beta$ be causal covectors. The bilinear form
\[
h=\alpha\otimes \beta+ \beta \otimes\alpha + g ,
\]
is positive definite iff
\begin{equation} \label{mis}
[1+g^{-1}(\alpha,\beta)]^2<g^{-1}(\alpha,\alpha)\, g^{-1}(\beta,\beta).
\end{equation}
Moreover, in this case $\alpha$ and $\beta$ are timelike and of the same time orientation.
\end{lemma}

\begin{proof}
The last statement is immediate from the dual of Lemma \ref{nis}, thus we need only to prove the equivalence between the positive definiteness of $h$ and Eq. (\ref{mis}).

To start with, we derive some results that are useful for both directions of the proof. Suppose that $\alpha,\beta\ne 0$ and that they are not proportional.
Let $\alpha^\sharp=g^{-1}(\alpha,\cdot)$, $\beta^\sharp=g^{-1}(\beta,\cdot)$, then the generic vector can be written
\begin{equation} \label{voq}
v=k+a \alpha^\sharp+b \beta^\sharp
\end{equation}
where $a,b\in \mathbb{R}$ and $k\in \textrm{ker} \alpha \cap  \textrm{ker} \beta$.
Indeed, the pair $(a,b)$ is uniquely determined by the equations
\begin{align*}
\alpha(v)&= a g^{-1}(\alpha , \alpha)+b g^{-1}(\alpha, \beta),\\
\beta(v)&=a g^{-1}(\alpha, \beta)+bg^{-1}(\beta,\beta),
\end{align*}
where the determinant is $g^{-1}(\alpha , \alpha)g^{-1}(\beta,\beta)-g^{-1}(\alpha, \beta)^2$ which is negative by the reverse Cauchy-Schwarz inequality.

We have setting $x=g^{-1}(\alpha, \alpha)$, $y=g^{-1}(\alpha,\beta)$, $z=g^{-1}(\beta,\beta)$
\begin{align} \label{vid}
h(v,v)=g(k,k)+a^2(2 x y+x )+b^2 (2 z y+z)+2ab(x z+y^2+y)
\end{align}
where $g(k,k)\ge 0$, because $k$, being orthogonal to the vectors $\alpha^\sharp\pm \beta^\sharp$, is spacelike (notice that $\alpha^\sharp$ and $\beta^\sharp$ are causal and not proportional, hence their sum or difference is timelike). The determinant of the quadratic form in $(a,b)$ appearing in (\ref{vid}) is
\begin{align}
x z(2y+1)^2-(x z+y^2+y)^2=(y^2-xz)[xz-(1+y)^2] .
\end{align}

$\Rightarrow$. Assume that $h$ is positive definite.
Clearly $\alpha\ne 0$ and $\beta\ne 0$, otherwise $h$ would be Lorentzian.

Let us suppose that $\alpha$ and $\beta$ are proportional in which case we can find a causal 1-form $\tau$ such that $\alpha=s \tau$, $\beta=\pm s^{-1} \tau$ for some $s>0$ and some sign, and hence $\alpha \otimes \beta =\pm \tau\otimes \tau$. Then, contracting the quadratic form $h=\pm 2\tau \otimes \tau +g$ with $\tau^\sharp$, we get
\begin{equation} \label{kkk}
h(\tau^\sharp,\tau^\sharp)=g^{-1}(\tau,\tau)[\pm 2g^{-1}(\tau,\tau)+1]
\end{equation}
which is not positive if $\tau$ is lightlike, thus $\alpha, \beta$ and $\tau$ are timelike. Moreover, under this condition we get positive definiteness of $h$  only if the  condition
\[
1\pm 2g^{-1}(\tau,\tau)<0
\]
holds. Now, it cannot be satisfied in the negative sign case, thus $\alpha$ and $\beta$ have the same time orientation  and we are left with precisely the condition given by Eq.\ (\ref{mis}).


Let us consider the case in which $\alpha$ and $\beta$ are not proportional.
A necessary condition for positive definiteness of $h$ is obtained from Eq.\ (\ref{vid}) setting $k=0$ and $b=0$ which gives $x\ne 0$, and hence $x<0$ (as $x\le 0$), and $2y+1<0$. Analogously another necessary condition is $z<0$, thus $\alpha$ and $\beta$ are timelike. Under this condition the positivity of $h(v,v)$ for $k=0$ implies  the positivity   of the  quadratic form (\ref{vid}) in $(a,b)$ and hence of its determinant. This condition reads
\begin{align}
0<(y^2-xz)[xz-(1+y)^2] .
\end{align}
The former factor on the right-hand side is positive by the reverse Cauchy-Schwarz inequality (since $y$ appears squared, this inequality holds independently of the relative time orientation of $\alpha$ and $\beta$, and equality holds only in the proportional case), thus we are left with $xz-(1+y)^2>0$ which is Eq.\ (\ref{mis}).

$\Leftarrow$. Assume that Eq.\ (\ref{mis}) holds true. By Lemma \ref{nis} the covectors $\alpha$ and $\beta$ are timelike and of the same time orientation, thus, keeping the previous notation, $x<0$, $z<0$.

If $\alpha$ and $\beta$ are proportional, we have $\alpha=s \tau$, $\beta=s^{-1}\tau$, for some timelike 1-form $\tau$ and real number $s>0$,  then Eq.\ (\ref{mis}) reads $1+2 g^{-1}(\tau,\tau)<0$, which due to the fact that $\tau$ is timelike, that $h(\tau^\sharp,\tau^\sharp)=g^{-1}(\tau,\tau)[ 2g^{-1}(\tau,\tau)+1]>0$, that $h(\tau^\sharp,\cdot)=[2g^{-1}(\tau,\tau)+1] \tau$, and $h \vert_{\textrm{ker} \tau}= g\vert_{\textrm{ker} \tau}$, implies that $h$ is positive definite.

Suppose that $\alpha$ and $\beta$ are not proportional.
The reverse Cauchy-Schwarz inequality and Eq.\ (\ref{mis}) give  $(1+y)^2<xz<y^2$ which implies $1+2y<0$, thus the quadratic form in $(a,b)$ in Eq.\ (\ref{vid}) is positive definite. As the generic vector $v$ can be written as in (\ref{voq}), the bilinear form $h$ in Eq.\ (\ref{vid}) is positive definite, and we have finished.
\end{proof}


We are ready to state our characterization of the closed Lorentzian submanifolds of Minkowski spacetime. We also state a number of properties that these spacetimes satisfy.

\begin{theorem} \label{mia}
Let $(M,g)$ be an $n$-dimensional spacetime $(M,g)$, where $g$ is $C^k$, $k\in \mathbb{N}\backslash\{1,2\} \cup \{\infty\}$, and let $s:=k$ for $k\ge 3$; $s:=1$ for $k=0$. The next properties are equivalent
\begin{itemize}
\item[(a)] there is  a $C^s$ proper isometric embedding in Minkowski spacetime $\mathbb{L}^{N+1}$ for some $N\ge 1$,
\item[(b)]      there are two $C^s$ (equiv.\ $C^1$, equiv.\ $C^\infty$) temporal functions $t_-,t_+\colon M\to \mathbb{R}$, such that
\begin{equation} \label{mgi}
\vert 1+g(\nabla t_-,\nabla t_+)\vert < \Vert \nabla t_-\Vert_g\, \Vert \nabla t_+\Vert_g,
\end{equation}
and for every $a,b\in \mathbb{R}$, $t_-^{-1}([a,+\infty))\cap t_+^{-1}((-\infty, b])$ is compact.
\item[(c)] On $M$ there is a $C^{k+1}$ (equiv.\ smooth) function $t\colon M\to \mathbb{R}$ such that
\begin{equation} \label{msu}
\tilde h:=g+2\dd t^2
\end{equation}
is a complete Riemannian metric on $M$.
\end{itemize}
In this case $N$ can be chosen to be $N(n,k)+1$, where $N(n,k)$ is the optimal value  for the isometric Riemannian problem.

Moreover, if these equivalent  conditions hold true then $(t_++t_-)/\sqrt{2}$ is temporal, steep and Cauchy and we have that
$\sqrt{2} t$ is temporal, strictly steep, $\tilde h$-steep,  Cauchy, and the level sets $S_a=t^{-1}(a)$ of $t$, endowed with the metric $\gamma^a$ induced  from $g$ are complete Riemannian manifolds.

Finally, the $C^s$ functions  $t_-,t_+$ and the smooth function $t$ can be found so as to be related as follows $t=(t_-+t_+)/2$ where $t_-$ and $t_+$ are both Cauchy, such that $t_-<t_+$,
and so that, defining $f=(t_--t_+)/2$, the function $t\times f\colon M\to \mathbb{R}^2$  is proper. Additionally, they can be chosen so that the following   properties hold true:
\begin{itemize}
\item[(i)] Every curve contained in $t^{-1}([a,b])$ for some $a,b \in \mathbb{R}$ and escaping every compact set has infinite length for $g+\dd t^2$.
\item[(ii)] For each $C^0$ function $\mu:M\to \mathbb{R}$  such that there is a constant $\epsilon>0$,  with $\mu>\epsilon$, we have that every $\tilde g$-spacelike curve, with $\tilde g=-\mu \dd t^2+g$,  escaping every compact set is complete (this is a kind of stable uniform spacelike completeness condition).
\item[(iii)] With the notations of the previous point: for $\mu$ bounded from above $t$ is still Cauchy in $(M,\tilde g)$.
\end{itemize}

\end{theorem}


\begin{proof}

Consider the smooth function  $F\colon \mathbb{R}^N\to [1/2,+\infty)$, $F=\frac{1}{2} \sqrt{1+\vec{X}\cdot \vec{X}}$.  The function $F$ is proper and   $\Vert \nabla^E F\Vert_E< 1/2$, thus it is 1/2-Lipschitz.



Regarding the coordinates of $\mathbb{R}^N$ as the last $N$-coordinates of $\mathbb{L}^{N+1}$, hence extending $F$ in the obvious way to  $\mathbb{L}^{N+1}$, we have that
\begin{align*}
H&:=\eta^{(N+1)}+2 (\dd X^0)^2 -2(\dd F)^2
=\delta^{(N+1)} - 2 (\dd F)^2
\end{align*}
is positive definite. Indeed, at those points where $\dd F=0$ this is clear, while if $\dd F\ne 0$, we have that in $\textrm{ker}\, \dd F$ it is positive definite. Evaluated once in $\nabla^E F$ it gives $(1-2\Vert \nabla^E F\Vert_E^2)\dd F$ which vanishes on $\textrm{ker} \, \dd F$, while evaluated twice it gives $ \Vert \nabla^E F\Vert_E^2(1-2\Vert \nabla^E F\Vert_E^2)\ge \frac{1}{2} \Vert \nabla^E F\Vert_E^2>0$ which proves that  $H$ is positive definite (observe that the inequality $\Vert \nabla^E F\Vert_E\le 1/2$ is non-ambiguous since independent of whether we are working in $E^N$ or $E^{N+1}$ as $F$ does not depend on $X^0$).

Let $V\in T\mathbb{L}^{N+1}$, since $\Vert \nabla ^E F\Vert_E\le 1/2$, we have $\Vert V\Vert_E^2=\Vert V\Vert_H^2+2(\p_{V} F)^2\le  \Vert V \Vert_H^2+\frac{1}{2} \Vert V\Vert_{E}^2$, thus $ \Vert V \Vert_{E} \le {\sqrt{2}}  \Vert V \Vert_{H}$, which implies  $d^E\le \sqrt{2} d^H$, and since the closed balls of $H$ with radius $r$ are contained in the compact balls of $\delta^{(N+1)}$ with radius $\sqrt{2} r$, by Hopf-Rinow the  metric $H$ is complete.

Let us assume $(a)$ and explore its consequences.
Suppose that there is a $C^s$ proper isometric embedding $\phi: M \to  \mathbb{L}^{N+1}$, $g=\phi^* \eta^{(N+1)}$.
Let $h:=\phi^*H$, since $d^H\circ (\phi \times \phi)\le d^h$ we conclude that $h$ is  a $C^{s-1}$ complete Riemannian metric.

Let us introduce the $C^s$ functions $t:=X^0\circ \phi$ and $f:=F\circ \phi$.
We have
\begin{align} \label{wlw}
h=&g+ 2\dd t^2-2\dd f^2.
\end{align}

Proof of (i). Let us
consider the degenerate metric
\begin{align*}
\check H&:=\eta^{(N+1)}+ (\dd X^0)^2=\sum_{i=1}^N (\dd X^i)^2.
\end{align*}
Every $C^1$ inextendible curve $\sigma\colon I\to M$ contained in $t^{-1}([a,b])$ for some $a,b\in \mathbb{R}$, $a\le b$, and escaping every compact set must have infinite length according to $\phi^*\check H=g+\dd t^2$, in fact this curve can be interpolated by  straight lines in $\mathbb{R}^{N+1}$ (hence not necessarily contained in $M$), so that the length for $\check H$ is not increased,
then the projection on the slice $X^0=0$ has infinite length because on that slice the induced metric from $\check H$ coincides with the Euclidean metric. Finally, the projection is non-expanding.

Proof of (ii).
Let   $\mu:M\to \mathbb{R}$ be a continuous function such that there is a constant $\epsilon>0$,  $\mu>\epsilon$, and let
 $\tilde g:=-\mu \dd t^2+g$, hence a Lorentzian metric on $M$ with cones wider than $g$.

Consider a $\tilde g$-spacelike curve $s \mapsto \gamma(s)$, so that evaluated twice on the tangent vector (we omit $ \gamma'$ for shortness) $-\mu \dd t^2+g\ge 0$, then over it for $k:=\frac{\mu}{2+\mu} \in (0,1)$, $\mu=\frac{2k}{1-k}$
\[
g=k g+(1-k) g\ge k g+(1-k)\mu \dd t^2=k(g+2 \dd t^2)\ge k h \ge \frac{\epsilon}{2+\epsilon} h ,
\]
 where $h$ is complete. This implies that if $\gamma$ escapes every compact set
 \[
 \int \sqrt{g(\gamma',\gamma')}\, \dd s \ge   \sqrt{ \frac{\epsilon}{2+\epsilon} } \int \sqrt{h(\gamma',\gamma')}\, \dd s=+\infty.
\]
Proof of $ (a)\Rightarrow (b)$. We have
\begin{align} \label{msd}
h=&g+ 2\dd t^2-2\dd f^2=g+\alpha\otimes \beta+\beta\otimes \alpha,
& \alpha =\dd (t-f), \quad \beta=\dd (t+f).
\end{align}
The functions $X_\pm=X^0\pm F$ are temporal in $\mathbb{L}^{N+1}$ because the gradient is $-\p_0\pm \nabla^E F$ and $F$ has Lipschitz constant smaller than 1. Thus the $C^s$ compositions
\[
t_\pm = t \pm f= (X^0\pm F)\circ \phi=X_\pm\circ \phi
\]
 are  temporal on $M$ (indeed $\dd X_+$ is positive over the future Minkowski causal cone and thus over its intersection with $\phi_*(TM)$, a fact that expresses the temporality of $t_+$), hence $t=(t_-+t_+)/2$ is temporal and $\dd t_\pm$ are timelike 1-forms. Equation (\ref{mgi}) now follows from Lemma \ref{llo}.

For  $a,b\in \mathbb{R}$, let us consider the closed set
\[
B=t_-^{-1}([a,+\infty))\cap t_+^{-1}((-\infty, b])
\]
Over it $2f\vert_B=(t_+-t_-)\vert_B \le b-a$. As $f$ is positive, it is bounded on $B$. Thus $t=t_-+f$ is lower bounded on $B$ and, as it can also be written  as $t=t_+-f$, it is also upper bounded on $B$. Since ${\vec{X}\cdot \vec{X}} =4F^2-1$, $B$ is contained in the $\phi$-inverse image of a compact cylinder of $\mathbb{L}^{N+1}$ which implies, by properness of $\phi$, that $B$ is compact in the topology of $M$.

Proof that $t$ is steep (the proof that the embedding can be chosen so that $t$ is smooth will be given later on). Notice that for every future causal vector $V$ on $\mathbb{L}^{N+1}$, $\dd X^0(V) \ge \sqrt{-\eta^{N+1}(V,V)}$, thus restricting to vectors in  $\phi_*(M)$ we get $\dd t(v)\ge   \Vert v\Vert_g$ for every future $g$-causal vector $v$, which implies the steepness condition $ \Vert \nabla^g t\Vert_g\ge 1$, as can be seen taking $v=-\nabla^g t$, see also \cite[Thm.\ 1.23]{minguzzi18b}.

Proof that $t$ is Cauchy. Since $\phi(M)$ is a closed subset of $\mathbb{L}^{n+1}$, a future (past) inextendible causal curve on $M$ is also future (resp.\ past) inextendible in $\mathbb{L}^{n+1}$. As a consequence, for each $a\in \mathbb{R}$, $\phi^{-1}((X^0)^{-1}(a)\cap \phi(M))=t^{-1}(a)$ is a Cauchy hypersurface, and so $t$ is Cauchy.

Proof that $t_-$ and $t_+$ are Cauchy. Let $V\in T\mathbb{L}^{N+1}$, $V=\beta \p_0+\alpha_i \p_i$, be future-directed causal, i.e.\ $\beta>0$, $\Vert \vec \alpha \Vert_E\le \beta$, then
\[
\dd X_+(V)=\beta+  \langle \nabla^E F, \vec \alpha \rangle\ge \beta-\vert \langle \nabla^E F, \alpha \rangle_E\vert\ge \beta - \frac{1}{2} \Vert \vec\alpha \Vert_E\ge \beta/2.
\]
 Every inextendible causal curve $\sigma\colon I\to M$ is inextendible in $\mathbb{L}^{N+1}$ and can be parametrized with $t=X^0\circ \sigma$, so that for $V=\phi_*(v)$, $v=\dot \sigma$, we have $\beta=1$ and the domain becomes $\mathbb{R}$ (because $t$ is Cauchy). The inequality $\dd X_+(\phi_*(\dot \sigma)) \ge 1/2$ proves that $t_+$ is Cauchy. Similarly, $t_-$ is Cauchy.

Moreover, by construction $f$ is proper on the subsets of the form $t^{-1}([a,b])$, as $F$ is proper on subsets of the form $(X^0)^{-1}([a,b])$. Thus $t\times f\colon M\to \mathbb{R}^2$, $p\mapsto (t(p),f(p))$, is proper.

Proof of $ (a)\Rightarrow (c)$ and properties of $t$. Equation (\ref{msu}) follows from Eq.\ (\ref{wlw}) setting $\tilde h=h+2\dd f^2$. Clearly, $\tilde h$ is complete as $h$ is.
Equation (\ref{msu}) evaluated  on future causal vectors  implies  that $\sqrt{2}\, t$ is strictly steep and $\tilde h$-steep, hence Cauchy. If $v$ is tangent to the spacelike hypersurface $S_a:=t^{-1}(a)$ then $\Vert v\Vert_{\tilde h} \le \sqrt{g(v,v)}=\sqrt{\gamma^a(v,v)}$ which implies the completeness of $(S_a,\gamma^a)$.

Proof of (iii). Let $\mu$ be a non-negative function bounded from above by a constant $q>0$.
The flat metric on $\mathbb{R}^{N+1}$, $\check \eta^{(N+1)}=-(1+q) (\dd X^0)^2+\sum_i (\dd X^i)^2$, induces the metric $\check g=- q \dd t^2 +g$ on $M$. 
  It is clear that $X^0$ is Cauchy for $\check \eta^{(N+1)}$ thus  $t$ is  Cauchy for $(M,\check g)$ (because every inextendible $\check g$-causal curve is $\check \eta$-causal inextendible ) and  hence $t$ is Cauchy for $(M,\tilde g)$.

$(b) \Rightarrow (a)$. The temporal conditions and Eq.\ (\ref{mgi}) are open conditions, thus if $t_-, t_+$ are $C^1$ and satisfy them, by the density argument based on  \cite[Thm.\ 2.6]{hirsch76} \cite[p. 35]{hirsch76} and already used in the proof of Thm.\ \ref{arg}, they can be found $C^\infty$ while still satisfying the properness condition in (b).

Suppose that there are two smooth temporal functions $t_-,t_+\colon M\to \mathbb{R}$ such that Eq.\ (\ref{mgi}) holds true. By Lemma \ref{llo} the $C^k$ metric on $M$
\[
h=\dd t_-\otimes \dd t_++ \dd t_+ \otimes\dd t_- + g= g+ 2\dd t^2-2\dd f^2 ,
\]
is positive definite, where we set $t=(t_-+t_+)/2$, $f=(t_+ - t_-)/2$. By Nash's theorem it can be $C^s$ isometrically embedded in $E^{N-1}$ for some $N\ge 2$. Let $\varphi\colon M\to E^{N-1}$ be the $C^s$ Nash embedding, and consider the $C^s$  embedding $\phi\colon M\to \mathbb{L}^{N+1}$ given by $\phi=(\sqrt{2} t, \varphi, \sqrt{2} f)$. Then
\[
\phi^*\eta^{(N+1)}=-(\dd (\sqrt{2} t))^2+h+(\dd (\sqrt{2} f))^2=g.
\]
Suppose that the embedding in $\mathbb{L}^{N+1}$ has a boundary point $q$ and let $\tilde O$ be a  relatively compact $\mathbb{L}^{N+1}$-open neighborhood of $q$, then on $M$ there is  a non-relatively compact open set $O=\phi^{-1}(\tilde O)$ over which $t$ and $f$ are bounded, which implies that $t_-$ and $t_+$ are bounded. In particular, $t_-\vert_O\ge a$, $t_+\vert_O\le b$, for some $a,b\in \mathbb{R}$, which due to $O\subset t_-^{-1}([a,+\infty))\cap t_+^{-1}((-\infty, b])$ gives a contradiction as the set on the right-hand side is compact.

This proves that the embedding is proper hence the desired implication, but it also proves that as $X^0\circ \phi=\sqrt{2} t$, the smooth function $\sqrt{2}t:=(t_++t_-)/\sqrt{2}$ is temporal, steep and Cauchy by the same argument used in the implication $(a) \Rightarrow (c)$. In fact,  one can redefine $f$, $t_-$ and $t_+$ as done there, which shows that $t$ can be chosen smooth while $t_-$, $t_+$ remain $C^s$.

$(c) \Rightarrow (a)$. The metric $\tilde h$ has regularity $C^{k}$.  Let $\Phi\colon M\to E^N$ be a $C^s$ proper isometric embedding of $(M,\tilde h)$ into $E^N$ for some $N$. It exists by Cor.\ \ref{odw}.
The $C^s$ embedding $\tilde\Phi\colon M\to \mathbb{L}^{N+1}$ given by $\tilde \Phi=(\sqrt{2} t ,\Phi)$, is isometric because
\[
\tilde \Phi^*\eta^{(N+1)}=-(\dd (\sqrt{2} t))^2+\tilde h= g.
\]
Equation (\ref{msu}) implies that $t$ is $\tilde h$-steep and the completeness of $\tilde h$ implies that $t$ is Cauchy. Suppose that the embedding $\tilde \Phi$ has a boundary point $Q$, then we can find $q_n\in M$, $\phi(q_n)\to Q$, thus  they escape every compact set on $M$. Let $\sigma_n$ be an inextendible causal curve passing through $q_n$, and let $r_n\in M$ be the intersection of $\sigma_n$ with $t^{-1}(t(q))= \phi^{-1}((X^0)^{-1}(X^0(Q)))$. Since $\sigma_n\circ \phi$ are also inextendible causal curves in $\mathbb{L}^{n+1}$ (remember that they have to intersects all the level sets of $X^0$, since $t$ is Cauchy) and since they pass through $\phi(q_n)$ that approach $Q$, the points $\phi(r_n)$ are such that $\phi(r_n)\to Q$. However, this fact implies that $\Phi$ is not closed, a contradiction.
\end{proof}

\begin{remark}
In \cite[Thm.\ 1]{muller13} (second statement) M\"uller recognized that if $(M,g)$ is closely isometrically embeddable then there is a smooth steep Cauchy temporal function $t$ with complete level sets, as we also stated in Thm.\ \ref{mia}. Actually, \cite[Thm.\ 1]{muller13} (first statement) claims the converse but the proof is incorrect.
The error,
 pointed
out by  Miguel S\'anchez and myself, was acknowledged in a  private communication (April 2017) by M\"uller.
He could amend the result by adding further assumptions such as mildness, cf.\ \cite{muller13} for the definition, but the erratum has yet to be published.

In general, it is false that if $t$ is a  smooth steep Cauchy temporal function  with complete level sets in the metric induced from $g$, then $\tilde h:=g+2 \dd t^2$ will be a complete Riemannian metric (so implying the closed embedding by point (c) of Thm.\ \ref{mia}). It is certainly Riemannian but not necessarily complete. Consider
$1+1$ Minkowski spacetime, $\eta=-\dd t^2+\dd x^2$, and let $\sigma\colon (0,\pi/2) \to M$, be the spacelike curve  $s \mapsto (s, \tan s)$. Let us consider a tubular neighborhood of $\gamma$ with compact constant-time sections, and let us make a non-trivial conformal rescaling $\eta \to g=\Omega \eta$, $\Omega\le 1$ just on it, with $\Omega$ approaching one at the boundary of the neighborhood.
Since  $\dd t(\dot \sigma)=1>0$
\[
\sqrt{\tilde h(\sigma',\sigma')}\le \sqrt{g(\dot\sigma,\dot \sigma)}+\sqrt{2} \, \dd t(\dot \sigma)
\]
and since  $t$ is bounded on $\sigma$, for a suitable choice of $\Omega$ we can bound the $\tilde h$-length of $\sigma$, which implies that $\tilde h$ is not complete as $\sigma$ escapes to infinity. However, the level sets of $t$ are  complete in the metric induced from $g$ since over each slice $t^{-1}(a)$, $a\in \mathbb{R}$, the metric has been modified just in a compact set. Finally, $t$ is still steep for $g$ because over a future $g$-causal (hence future $\eta$-causal) vector $v$, $\dd t(v)\ge \Vert v\Vert_\eta \ge \Vert v\Vert_g$ as $\Omega\le 1$.
\end{remark}

\begin{remark}[Geometrical meaning of the steep factor] Let us regard the embedded manifold $M$ as a timelike submanifold of $\mathbb{L}^{N+1}$, and let $t$ be the steep functions given by the restriction of the first coordinate.
The vector $TM\ni w=\nabla t/g(\nabla t,\nabla t)$ satisfies $\dd t (w)=1$ and $g(w,w)=1/g(\nabla t,\nabla t)$.
The vector $w$  can be decomposed as follows $w=\p_0+\alpha u$, $\alpha \ge 0$, where $u\in T\mathbb{L}^{N+1}$ satisfies $\dd X^0(u)=0$ and is normalized according to the metric $\sum_i (\dd X^i)^2$.
Thus $g(w,w)=\eta^{N+1}(w,w)=-1+\alpha^2$ from which we get $\alpha=\sqrt{1-\frac{1}{-g(\nabla t,\nabla t)}}$.
Every spacelike vector in $TM$ orthogonal to $\nabla t$, belongs to $\textrm{ker} \dd X^0$, so it is also orthogonal (in both the Euclidean and Lorentzian interpretations) to $\p_0$ and hence $u$.
Notice that $\alpha$ is the tangent to the Euclidean angle formed by $w$ and $\p_0$ so it represents how much $TM$ is tilted with respect to $\p_0$. Thus the closer the length of $\nabla t$ is to 1, the smaller the Euclidean angle between $TM$ and $\p_0$.
\end{remark}

\begin{corollary}
Let $(M,g)$ be a globally hyperbolic spacetime admitting a compact Cauchy hypersurface, and let $g\in C^k$, $k\in \mathbb{N}\backslash\{1,2\} \cup \{\infty\}$. Then, there is a $C^s$ proper isometric embedding in Minkowski spacetime $\mathbb{L}^{N(n,k)+2}$, where  $s:=k$ for $k\ge 3$; $s:=1$ for $k=0$.
\end{corollary}

\begin{proof}
Every globally hyperbolic spacetime with $g\in C^0$ admits a smooth steep temporal function $t$ \cite[Thm.\ 3.12]{minguzzi17}, thus $\Vert \nabla t \Vert_g\ge 1$. The spacetime is homeomorphic to the product $\mathbb{R}\times S$ where $S$ is homeomorphic to the Cauchy hypersurface (any two Cauchy hypersurfaces are homeomorphic), and where the projection on the first factor is $t$. Let $t_-:=t=:t_+$, then the inequality (\ref{mgi}) is satisfied and  since $t^{-1}([a,b])$ is homeomorphic to $[a,b]\times S$, it is compact, which concludes the proof.
\end{proof}

\begin{definition}
A {\em proper} spacetime is a spacetime that satisfies the equivalent properties of Thm.\ \ref{mia}.
\end{definition}

%
%

The following general result can be applied to the existence of the isometric embedding case, and gives information on the form of the metric in the open set, possibly the whole spacetime $M$, where the gradients of $t_-$ and $t_+$ are not proportional.

\begin{proposition} \label{pro}
Let $t_-$ and $t_+$ be $C^1$ functions with causal non-proportional gradients on an open region of a spacetime $(M,g)$,  $g\in C^0$. Then there is a $C^0$ Riemannian metric $m_{c_-,c_+}$ on the condimension 2  manifolds $\Sigma_{c_-,c_+}$ intersection of the spacelike hypersurfaces $t_-=c_-$, $t_+=c_+$ for  constants $c_-, c_+$, in such a way that
\begin{equation} \label{vow}
g=-[\alpha \dd t_-^2+2\beta \dd t_- \dd t_+ +\gamma \dd t_+^2]\oplus m_{t_-,t_+}
\end{equation}
for some $C^0$ functions $\alpha, \beta, \gamma:M\to \mathbb{R}$, where  $\alpha,\gamma \le 0$.
Conversely, if  (\ref{vow}) is a $C^0$ Lorentzian metric for some $C^1$ functions $t_-,t_+$, $C^0$ functions $\alpha,\beta,\gamma$, and a $C^0$ Riemannian metric $m_{t_-,t_+}$, and if the inequalities 
$\alpha,\gamma\le 0$  hold, then  $t_-$ and $t_+$ have causal non-proportional gradients.

Moreover, in this case ${\beta^2-\alpha \gamma}> 0$, $\beta\ne 0$, and we have the identities
\begin{align*}
\alpha&= -\frac{\Vert \nabla t_-\Vert^2_g}{\Delta^2}, \quad
&\beta&= \frac{-g(\nabla t_-, \nabla t_+)}{\Delta^2}, \quad
&\gamma&= -\frac{\Vert \nabla t_+\Vert^2_g}{\Delta^2},
\end{align*}
where $\Delta^2:=g(\nabla t_-,\nabla t_+)^2-\Vert \nabla t_+\Vert^2_g \Vert \nabla t_-\Vert^2_g=\frac{1}{\beta^2-\alpha \gamma}>0$. The gradients have the same time orientation iff $\beta>0$.

In particular,  the  inequality
\begin{equation} \label{nine}
\vert 1+g(\nabla t_+,\nabla t_-)\vert <  \Vert \nabla t_+\Vert_g\, \Vert \nabla t_-\Vert_g,
\end{equation}
is equivalent to the  inequality (actually also in the non-strict case)
\begin{equation} \label{ten}
(\beta-1)^2 < \alpha \gamma  ,
\end{equation}
and to the fact that  $g +2 \dd t_- \dd t_+$ is positive definite, in which case $\nabla t_-$ and $\nabla t_+$ are timelike with the same time orientation.
\end{proposition}

Observe that under the proper isometric embedding assumptions of Theorem \ref{mia} (see item (b)), the locus $t_-^{-1}(c_-)\cap  t_+^{-1}(c_+)$ is compact. If $\dd t_-$ and $\dd t_+$ are not proportional over it, then it is a compact spacelike submanifold.

\begin{proof}
Assume that $t_-$ and $t_+$ have past-directed causal non-proportional gradients.
Let us consider the metric
\[
h=g+\alpha \dd t_-^2+2\beta \dd t_- \dd t_+ +\gamma \dd t_+^2 .
\]
We want to find functions $\alpha, \beta, \gamma$ in such a way that $h$ is annihilated by any vector in $\textrm{span} (\nabla t_-, \nabla t_+)$. Note that (\ref{vow}) can be written, with $m=h\vert_{\ker \dd t_-\cap \ker \dd t_+}$, iff these functions exist.
Let us calculate (we write for shortness $\Vert \nabla t\Vert^2_g=-g(\nabla t,\nabla t)=-g^{-1}(\dd t, \dd t)$ for a function with causal gradient).
\begin{align*}
h(\nabla t_-,\cdot)
&=[1-\alpha \Vert \nabla t_-\Vert^2_g + \beta g(\nabla t_-,\nabla t_+)] \dd t_-+[-\beta  \Vert \nabla t_-\Vert^2_g+\gamma g(\nabla t_-,\nabla t_+) ] \dd t_+
\end{align*}
Since this quantity, and the analogous quantity for $t_+$, must vanish, we have
\begin{align}
0&=1-\alpha \Vert \nabla t_-\Vert^2_g + \beta g(\nabla t_-,\nabla t_+), \label{one}\\
0&=1-\gamma \Vert \nabla t_+\Vert^2_g + \beta g(\nabla t_-,\nabla t_+),\\
0&=-\beta  \Vert \nabla t_-\Vert^2_g+\gamma g(\nabla t_-,\nabla t_+),\\
0&=-\beta  \Vert \nabla t_+\Vert^2_g+\alpha g(\nabla t_-,\nabla t_+) \label{four}
\end{align}

Observe that $\Delta^2>0$ by the reverse triangle inequality, and $g(\nabla t_-,\nabla t_+)\ne 0$ as these gradients are not proportional.
Multiplying the first equation by $\Vert \nabla t_+\Vert^2$ and using the last equation, we obtain
$\alpha \Delta^2=-\Vert \nabla t_+\Vert^2_g$ and analogously, $\gamma \Delta^2=-\Vert \nabla t_-\Vert^2_g$. Finally,  multiplying the first equation by $\Delta^2$ we get
$\beta \Delta^2=-g(\nabla t_-,\nabla t_+)$.  The identity  $(\beta^2-\alpha \gamma)\Delta^2=1$ follows easily. Note that the causality of $\nabla t_-$ and $\nabla t_+$ is related to the signs of $\alpha$ and $\gamma$ respectively, and that they have the same time orientation iff $-g(\nabla t_-,\nabla t_+)>0$ i.e. $\beta>0$. The  statement on the equivalence between (\ref{nine}) and (\ref{ten}) is now proved with a trivial calculation, while that on the equivalence between (\ref{nine}) and the positive definiteness of $g+\dd t_-\otimes \dd t_++\dd t_+\otimes \dd t_-$ follows from Lemma
\ref{llo}.

For the converse, suppose that Eq.\ (\ref{vow}) holds with $g$ Lorentzian and $\alpha,\gamma\le 0$.  The metric in square brackets must be Lorentzian, that is, its determinant must be negative, which gives $\alpha \gamma-\beta^2<0$, and hence $\beta\ne 0$. Applying  Eq.\ (\ref{vow}) to $\nabla t_-$ and then again to $\nabla t_+$ we get again Eqs.\ (\ref{one})-(\ref{four}). From the last two equations
\[
\beta^2 \Vert \nabla t_-\Vert ^2_g \Vert \nabla t_+\Vert ^2_g = \alpha \gamma g(\nabla t_-, \nabla t_+)^2<\beta ^2 g(\nabla t_-, \nabla t_+)^2
\]
and hence $\Delta^2>0$ and $g(\nabla t_-, \nabla t_+)\ne 0$. Following the same steps as before, $0\le -\alpha \Delta^2=\Vert \nabla t_+\Vert^2_g$ which implies that $\nabla t_+$ is causal, and similarly $\nabla t_-$ is causal. Now, $\Delta^2>0$ means that they are not proportional.
%
%
\end{proof}


\section{Conformal embeddings}

The goal of this section is to prove the following theorem which establishes that global hyperbolicity characterizes the causal structure of closed Lorentzian submanifolds of Minkowski spacetime.

Remember that $(M,g)$, $g\in C^k$, is $C^s$ conformally embeddable if there is a $C^k$ function $\Omega^2:M\to (0,\infty)$ such that $\tilde g:=\Omega^2 g=\phi^* \eta^{(N+1)}$, where $\phi$ is the $C^s$ embedding.

\begin{theorem} \label{kdp}
Let $(M,g)$ be an $n$-dimensional spacetime $(M,g)$, where $g$ is $C^k$, $k\in \mathbb{N}\backslash\{1,2\} \cup \{\infty\}$, and let $s:=k$ for $k\ge 3$; $s:=1$ for $k=0$.

The spacetime $(M,g)$ can be $C^s$ properly conformally embedded in Minkowski spacetime $\mathbb{L}^{N+1}$ for some $N>0$ iff it is globally hyperbolic (in other words the globally hyperbolic spacetimes are precisely the spacetimes that are conformally related to proper spacetimes). Moreover, for $k\ge 3$, if $S$ is a $C^{k+1}$ spacelike Cauchy hypersurface a conformal $C^k$ embedding $\phi\colon M\to \mathbb{L}^{N+1}$ can be  chosen such that $(\phi^0)^{-1}(0)=S$.
The integer $N$ can be chosen to be  $N(n,k)$, where $N(n,k)$ is the optimal value  for the isometric Riemannian problem.
\end{theorem}

\begin{proof}
$\Rightarrow$.
 Let $\tilde g$ be conformally related to $g$, and let $(M,\tilde g)$ be a closed Lorentzian submanifold of $\mathbb{L}^{N+1}$ (for simplicity we omit the embedding $\phi$).
Every inextendible causal curve in  $M$ is inextendible and causal in $\mathbb{L}^{N+1}$, thus they have to intersect the set $(X^0)^{-1}(0)\cap M$ which is then a Cauchy hypersurface for $(M,\tilde g)$ and hence for $(M,g)$.

\begin{figure}[ht]
\centering
\includegraphics[width=9cm]{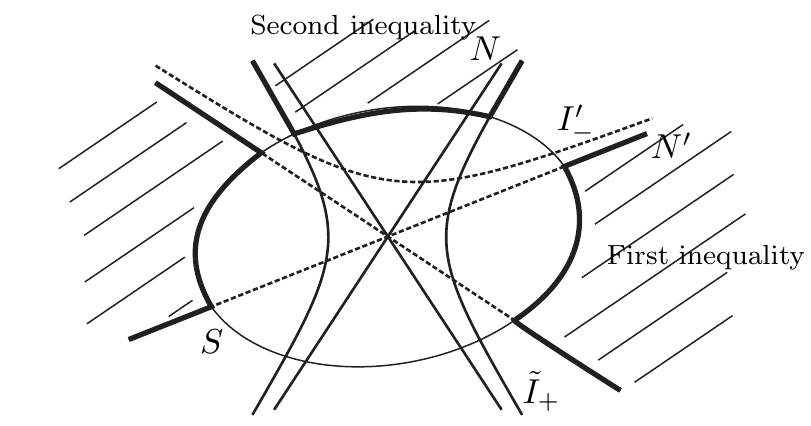}
\caption{The geometrical intuition behind the inequalities in Thm.\ \ref{kdp}. Here $S$ is the unit ball of $h$, $I_\pm=\{v\in T_pM\colon g(v,v)=\pm 1\}$ and similarly for $g'$, while $N$ and $N'$ are the null cones.  First, given $S$, $N$ and $N'$, rescale $g$ into $\tilde g$ so that $\tilde I_+$ does not contain $g'$-spacelike or $g'$-lightlike vectors of $h$-norm one. This gives the first inequality. Then rescale $g'$ so that $I'_-$ does not contain vectors $v$ such that $\tilde g(v,v)\le 1$ and $h(v,v)\ge 1$, this gives the second inequality.}
\end{figure}

$\Leftarrow$. Since global hyperbolicity is stable \cite{minguzzi11e} \cite{fathi12} \cite[Thm.\ 2.91]{minguzzi17} we can find $g'>g$ such that $(M,g')$ is globally hyperbolic and (for $k\ge 3$) $S$ is a Cauchy hypersurface for $(M,g')$.
Let $h$ be a complete Riemannian metric on $M$ and let $p\in M$. There are  sufficiently large $\lambda_1(p)>0$ such that for every $v\in T_pM$,
\begin{align*}
 g'(v,v) \ge 0 \ \textrm{ and } \  h(v,v)\ge 1 \  &\Rightarrow \  \lambda_1 g(v,v)> 1,
\end{align*}
Since the inequality on the right-hand side is strict, the same implication for the same value of $\lambda_1$ holds in a neighborhood $V_p$ of $p$.
Let us introduce a locally finite refinement $\{U_i\}$ and a partition of unity $\{\varphi_i\}$, and let us define the smooth positive function $\Omega_1:=\sum_i \lambda^i_1 \varphi_i$. Defining  $\tilde g:=\Omega_1 g$, we get for every $v\in TM$
\begin{align*}
 g'(v,v) \ge 0 \ \textrm{ and } \  h(v,v)\ge 1 \  &\Rightarrow \  \tilde g(v,v)> 1, \qquad (\textrm{first inequality}) .
\end{align*}
Similarly, for each $p$ and sufficiently large $\lambda_2(p)>0$ we have for  every $v\in T_pM$,
\begin{align*}
\tilde g(v,v)\le 1 \ \textrm{ and } \ h(v,v)\ge 1 \  &\Rightarrow \ \lambda_2 g'(v,v)<-1.
\end{align*}
Since the inequality on the right-hand side is strict, the same implication for the same value of $\lambda_2$ holds in a neighborhood $W_p$ of $p$.
By introducing a  partition of unity $\{\hat\varphi_i\}$, the function  $\Omega_2:=\sum_i \lambda^i_2\hat \varphi_i$, we find that for every $v\in TM$
\begin{align*}
\tilde g(v,v)\le 1 \ \textrm{ and } \ h(v,v)\ge 1 \  &\Rightarrow \ g'(v,v)<-1, \qquad (\textrm{second inequality}) ,
\end{align*}
where we redefined $\Omega_2 g'\to g'$ (it is still globally hyperbolic and such that $g'>\tilde g$).
%
As $(M,g')$ is globally hyperbolic it admits a smooth $g'$-temporal function  which is $g'$-steep, i.e.\ for every future-directed $g'$-causal vector $v$, $\dd t(v)\ge \sqrt{-g'(v,v)}$ \cite[Thm.\ 3.12]{minguzzi17}. 
Moreover, for $k\ge 3$, as $S$ is a $C^{k+1}$ spacelike 
Cauchy hypersurface for $(M,g')$, by \cite[Thm.\ 1]{bernard19} \cite[Thm.\ 2.5]{minguzzi19d} we can find $t$ as above but of regularity $C^{k+1}$ where additionally  $S=t^{-1}(0)$.

Let us prove that the $C^k$ metric $\tilde h=\tilde g+2 \dd t^2$ is a complete Riemannian metric by showing that $\tilde h\ge h$. Let $v$ be such that $h(v,v)=1$, if $\tilde{g}(v,v)\ge 1$ then clearly  $\tilde h(v,v)\ge 1$ and we have finished, if instead $\tilde{g}(v,v)< 1$ then from the second implication in display $g'(v,v)<-1$, in particular $v$ is $g'$-timelike, thus by $g'$-steepness of $t$,  $[\dd t(v)]^2\ge {-g'(v,v)}>1$ and hence $\tilde h(v,v)\ge 1$. By characterization (c) of Thm.\ \ref{mia}, $(M,\tilde g)$ is $C^s$ properly isometrically embeddable in Minkowski spacetime, thus $(M,g)$ is $C^s$ properly conformally embeddable in Minkowski spacetime. The proof of
$(c) \Rightarrow (a)$ in Thm.\ \ref{mia} shows that the zero coordinate of the embedding reads $\sqrt{2} t$ from which the last statement follows.
\end{proof}%
%

\renewcommand{\arraystretch}{1.5}

\begin{figure}
\begin{tabular}{ r|c|c| }
\multicolumn{1}{r}{}
 &  \multicolumn{1}{c}{Non-closed}
 & \multicolumn{1}{c}{Closed} \\
\cline{2-3}
Conformal & Stably causal   & Globally hyperbolic     \\
 &  (Thm.\ \ref{bys2})  &   (Thm.\ \ref{kdp}) \\
\cline{2-3}
Isometric &  Stable    & Proper \\
 &   (Thm.\ \ref{bys})   & (Thm.\ \ref{mia}) \\
\cline{2-3}
Isometric for   &  Globally hyperbolic   & None   \\
every $g' \in [g]$ &  (Rem.\ \ref{jvo})  &  (Rem.\ \ref{dob})  \\
\cline{2-3}
\end{tabular}
\caption{Equivalence between type of embedding and spacetime property. Here ``non-closed'' means that the embedding need not be closed. The second column is filled by the results of this work. Notice that: Proper $\Rightarrow$ globally hyperbolic $\Rightarrow$  stable $\Rightarrow$  stably causal. }
\end{figure}

%

\begin{remark} \label{dob}
The question analogous to that of Remark \ref{jvo} but for closed embeddings is: What are the conformal structures $(M,[g])$ such that for every $\tilde g\in [g]$, $(M,\tilde g)$ is closely isometrically embeddable in Minkowski spacetime?

It has to be globally hyperbolic so for some $g$  in the conformal class it is diffeomorphic to $\mathbb{R}\times S$, with metric $g=-\beta \dd t^2+h_t$ \cite{bernal04}. However, we can also consider the representative of the conformal class obtained multiplying by $\beta^{-1}$, i.e.\ $\tilde g=-\dd t^2+\tilde h_t$. Let $\varphi\colon M\to (0,\infty)$ be such that for every $x\in S$, $\varphi(\cdot,x)$ reaches a maximum at $t=0$ and in the sense of bilinear forms $\varphi_t h_t\le 2 \varphi_0 h_0$, where $\varphi <1$ is so small for $t=0$, that $ (S,2\varphi_0 h_0)$ is incomplete. Then, since every Cauchy hypersurface on $M$ is a graph over $S$, $(M,\hat g)$, with $\hat g= \varphi \tilde g$, has the property that the projection to $(S, 2 \varphi_0 h_0)$ of any Cauchy hypersurface (endowed with the induced metric) is non-contracting thus every Cauchy hypersurface is incomplete in the metric induced from $\hat g$.

However, any closely isometrically embeddable spacetime has a foliation of Cauchy hypersurfaces that are complete in the induced Riemannian metric. This shows that $(M,\hat g)$ is not closely isometrically embeddable. In conclusion the answer to the initial question is: none.
\end{remark}

\section{Some consequences of the embedding existence} \label{bjd}

Thanks to the embedding it is possible to obtain a number of interesting results on the representation of local (or pointwise) properties by means of special time functions. Observe that the latter objects express global properties of the spacetime, so this is a kind of local-global relationship.

In the following result by `smooth' we really mean $C^s$ where $s$ is the constant expressing the regularity of the embedding as introduced in the previous theorems.
\begin{proposition}
Let $(M,g)$ be a stably causal spacetime and let $v\in T_pM$ be  a past-directed lightlike vector. Then there is a smooth  function $f\colon M\to \mathbb{R}$, having past-directed causal  gradient  $\nabla f$ such that $\nabla f(p)=v$.
\end{proposition}

\begin{proof}
We know that there is a function $\Omega^2:M\to \mathbb{R}$ such that  $(M,\Omega^2 g)$ can be regarded as a closed Lorentzian submanifold of $\mathbb{L}^{N+1}$. The vector $v$ reads  $v= a (-\p_0+{ u})$, $a>0$, where $\dd X^0(u)=0$, and $\Vert u \Vert_E= 1$. But then the function $F=X^0+\sum_{i=1}^N X^i u_i$ is  such that $\dd F$ is non-negative on any future causal cone of $\mathbb{L}^{N+1}$ and such that $ \dd F(v)=0$ at $p$, thus $f=b F\vert_M$, for some constant $b>0$, has the required properties. Indeed, $\dd f$ is non-negative over the future causal cones (which are the same for $g$ and $\Omega^2 g$), which implies that $\nabla^g f$ is past-directed causal. But the equality case of the reverse triangle inequality and $\dd f(v)=0$ at $p$ imply that $\nabla^g f$ is proportional to $v$ at $p$, equality being obtained adjusting $b$.
\end{proof}

\begin{proposition} \label{vix}
Let $(M,g)$ be a stable spacetime and let $v\in T_pM$ be a unit past-directed timelike vector. Then there is a smooth  steep temporal function $f\colon M\to \mathbb{R}$, such that  $\nabla f(p)=v$. If the spacetime is proper then $f$ can be chosen Cauchy.
\end{proposition}

\begin{proof}
The spacetime $(M,g)$ can be regarded as a Lorentzian submanifold of $\mathbb{L}^{N+1}$. The vector $v$ reads  $v= a (-\p_0+{ u})$, $a>0$, where $\dd X^0(u)=0$, and $\Vert u \Vert_E<1$. We can actually change the coordinates of $\mathbb{L}^{N+1}$ through a Lorentz transformation preserving the time orientation, in such a way that, by using the new coordinates, we have the same expressions as above but with $u=0$, $a=1$. Then, still using the new $X^0$, $f=X^0\vert_M$ provides the smooth steep temporal function. Observe that at $p$, $g(\nabla f, v)=\dd f(v)=\dd X^0(-\p_0)=-1$. But the reverse triangle inequality is $-g(\nabla f, v)\ge \Vert \nabla f\Vert_g \Vert v\Vert_g$. Since by the steepness condition $\Vert \nabla f\Vert_g \ge 1$, and by assumumption $\Vert v\Vert_g=1$,  we are in the equality case of the reverse triangle inequality, which implies $\nabla f \propto v$, and moreover $\Vert \nabla f\Vert_g = 1$, which imply $\nabla f=v$.
We already observed in previous proofs that if the spacetime is proper the function $f=X^0\vert_M$ is Cauchy as every inextendible casual curve in $M$ is inextendible in $\mathbb{L}^{N+1}$.
\end{proof}

We have the following interesting representation formula for the metric in terms of steep temporal functions. It is a kind of infinitesimal analog of the distance formula \cite[Thm.\ 4.6]{minguzzi17}.

%
%

\begin{theorem}
Let $(M,g)$ be a stable spacetime. A vector $v\in TM$ is timelike iff $\inf_f \vert\dd f(v)\vert^2\ne 0$, where $f$ runs over the smooth steep temporal functions, and in this case
\begin{equation} \label{oet}
g(v,v)=-\inf_f \vert\dd f(v)\vert^2,
\end{equation}
in fact the infimum is attained. Moreover, for $v$ spacelike we have $\inf_f \dd f(v)=-\infty$. If the spacetime is proper then the functions $f$ can also be demanded to be Cauchy.
\end{theorem}

\begin{proof}
The spacetime $(M,g)$ can be regarded as a Lorentzian submanifold of $\mathbb{L}^{N+1}$. The vector $v\in TM$ can be regarded as an element of $T\mathbb{L}^{N+1}$. If it is spacelike, we can choose the canonical coordinates of $\mathbb{L}^{N+1}$ in such a way that $\dd X^0(v)=0$. Then $f=X^0\vert_M$ provides a smooth steep temporal function such that $\dd f(v)=0$. Thus if $v$ is a spacelike vector $\inf_f \vert\dd f(v)\vert^2= 0$.

For $v$ spacelike we can also choose canonical coordinates in $\mathbb{L}^{N+1}$ in such a way that, in Euclidean terms, $v$ is as long as desired and close to the past lightlike cone. This is accomplished by boosting in the `opposite direction' of $v$ (in this qualitative description I assume some familiarity of the reader with the Lorentz transformations). As a result $\dd f(v)=\dd X^0(v)$ can be made as negative as desired.

If $v$ is a lightlike vector then $v=a(\p_0+u)$, $a>0$, $\Vert u\Vert_E=1$. But canonical coordinates on $\mathbb{L}^{N+1}$ can be chosen in such a way that $a$ is as small as desired. For any choice of canonical coordinates $f=X^0\vert_M$ provides a  smooth steep temporal function such that $\dd f(v)=a \dd X^0(-\p_0)=-a$, $\vert \dd f(v)\vert=a$, which proves that $\inf_f \vert\dd f(v)\vert^2= 0$.

Let us prove formula (\ref{oet}) for a timelike vector, this will also prove the first statement of the theorem. It is sufficient to prove it for $v$ future-directed, as the expressions on both sides do not depend on the time orientation.

If $f$ is steep and $v$ is future-directed timelike  then $\dd f (v) \ge \Vert v\Vert_g$ which reads $\vert \dd f(v)\vert^2\ge - g(v,v)$. Thus  we have the inequality $-g(v,v) \le  \inf_f \vert\dd f(v)\vert^2$.

Observe that $-v/\Vert v\Vert_g$ is a unit past-directed  timelike vector thus, by Prop.\ \ref{vix}, there is a smooth steep temporal function $f$ such that $\nabla f= -v/\Vert v\Vert_g$, which implies $\dd f(v)=g(\nabla f,v)=-g(v,v)/\Vert v\Vert_g=\Vert v\Vert_g$, hence formula (\ref{oet}) where the infimum is attained. The proof of the last statement goes as in the previous proofs.
\end{proof}

\section{Conclusions}

We proved that the globally hyperbolic conformal structures are precisely the conformal structures of timelike closed submanifolds of Minkowski spacetime (Thm.\ \ref{kdp}).

We considered the problem of characterizing the Lorentzian manifolds closely isometrically embeddable in Minkowski spacetime by means of intrinsic properties.
We found two possible characterizations and obtained further properties of these spacetimes that might lead to further characterizations (Thm.\ \ref{mia}). 
We also studied the form of the metric for these spaces (Prop.\ \ref{pro}).

It is interesting to observe that these spacetimes, which we called proper, are even nicer than globally hyperbolic spacetimes. It is expected that they might play some role in Physics, though this possibility has yet to be explored and clarified.

The advantage of the closed conformal  embedding with respect non-closed ones is that inextendible causal curves are mapped to inextendible causal curves, which often allows one to use causality results on $\mathbb{L}^{N+1}$ to infer results on $(M,g)$. The causal boundary of $(M,g)$ can be identified with a subset of the Penrose's conformal boundary of Minkowski spacetime which might also lead to interesting simplifications.

In general, the solution to the embedding problem might lead to  new strategies to attack causality and metric problems for spacetimes.
For instance, some local-global results can  be easily obtained via the embedding. We provided an example in Sec.\ \ref{bjd} where we proved an infinitesimal distance formula.


\end{document}